\documentclass[11pt,a4paper]{article}
\usepackage[T1]{fontenc}
\usepackage[utf8]{inputenc}
\usepackage{lmodern}
\usepackage{geometry}
\usepackage{amsmath,amssymb,amsthm,mathtools}
\usepackage{booktabs}
\usepackage{microtype}
\usepackage[colorlinks=true,linkcolor=blue,citecolor=blue,urlcolor=blue]{hyperref}
\usepackage{enumitem}
\numberwithin{equation}{section}

\newtheorem{theorem}{Theorem}[section]
\newtheorem{lemma}[theorem]{Lemma}
\newtheorem{corollary}[theorem]{Corollary}

\theoremstyle{remark}
\newtheorem{remark}[theorem]{Remark}

\newcommand{\Tr}{\operatorname{Tr}}
\newcommand{\trank}{\operatorname{trank}}
\newcommand{\hdim}{\operatorname{dim}}
\newcommand{\VC}{\operatorname{VCdim}}
\newcommand{\poly}{\operatorname{poly}}
\newcommand{\SAT}{\textsc{SAT}}
\newcommand{\MMHS}{\textsc{Maximum Minimal Hitting Set}}

\newcommand{\cG}{\mathcal{G}}
\newcommand{\cH}{\mathcal{H}}
\newcommand{\cL}{\mathcal{L}}

\title{Lower Bounds for Minimal-Transversal Enumeration\\
       of Bounded-Dimension Hypergraphs and the maximum Minimal Hitting Set Problem}
\author{Arnaud Mary}
\date{October 2026}

\begin{document}
\maketitle

\begin{abstract}

We present lower bounds for the maximum minimal hitting set problem and for the problem of
enumeration of minimal transversals in bounded-dimension hypergraphs.
Our starting point is the recent reduction of Kobayashi,
Kurita, and Wasa (2026), developed to establish ETH-based hardness of
hypergraph dualization. Adapting their construction, we show that,
under ETH, deciding whether a hypergraph $\cG$ with $n$ vertices and $m$ 
hyperedges has a minimal transversal of size at least $k$
cannot be done in time $n^{f(k)}m^{o(k)}$, for an arbitrary computable $f$.
We transfer this bound to the following decision problem: given a
dimension-at-most-$r$ hypergraph and a family of $s$ of its minimal
transversals, decide whether another exists. Under ETH, this problem
admits no uniform algorithm with running time $n^{f(r)}(s+1)^{o(r)}$.
Under SETH, for every fixed $t \geq 2$, it admits no algorithm with running
time $\poly(|\cG|)(s+1)^c$ for dimension at most $2t$ and any $c<t$.
As corollaries, the same exponent obstructions hold for the total time
of any output-polynomial enumeration algorithm, with $s$ replaced by the
total number of minimal transversals. In particular,
polynomial delay is excluded under SETH already for dimension at most four.

\end{abstract}

\section{Introduction}

Given a hypergraph $\cH$, the problem of enumerating its minimal
transversals, also called minimal hitting sets, is equivalent to dualization
of monotone Boolean functions. Its general output-polynomial complexity has
been open for long time. The classical duality algorithm of Fredman and Khachiyan
runs in quasipolynomial time
\cite{FredmanKhachiyan1996,EiterMakinoGottlob2008}. Very recently,
Kobayashi, Kurita, and Wasa \cite{KKW2026} proved  that that under the Exponential Time Hypothesis (ETH),
no output-polynomial time algorithm exists to enumerating minimal transversals in general hypergraphs. 

In this paper we prove that their reduction is flexible enough to yield lower bounds for the enumeration of minimal transversals in 
hypergraphs of bounded dimension, i.e. in hypergraphs where the maximum size of any hyperedge is bounded.

For bounded input dimension $k$, minimal transversals admit incremental-polynomial
enumeration \cite{EiterGottlob1995,EiterGottlobMakino2003}. Similarly, the work in
\cite{Mary2026} established incremental-polynomial enumeration for bounded
VC-dimension, unifying earlier tractable classes, including
bounded dimension, bounded hyperedge intersections, and several sparse and geometric
families \cite{KhachiyanEtAl2007}. The exponent of the number of generated
solutions in these algorithms depends on the dimension or the VC-dimension.
Whether polynomial delay can be achieved for each fixed dimension $k \geq 3$
was open.

Independently of enumeration problems, we consider the problem of deciding whether a hypergraph $\cH$ has a minimal transversal of size at least $k$.
This problem, called the \MMHS  problem (MMHS), or sometimes the \emph{upper transversal problem} or \emph{transversal rank problem}, 
is known to be NP-hard and W[1]-hard when parameterized by $k$ 
More recent work in \cite{AraujoBougeretCamposSau2023} 
show that under ETH, MMHS cannot be solved in time $(n+m)^{o(k)}$ where $n$ is the number of vertices, $m$ is the number of hyperedges.
. 

We establish two types of conditional lower bounds.

\begin{enumerate}
  
\item We first prove under ETH that the maximum minimal hitting set problem cannot be solved in time $n^{f(k)}m^{o(k)}$ for any computable function $f$. 
Using the small-transversal construction underlying the equivalences
in \cite{Schirneck2026}, we obtain a lower bound for the associated
decision problem: given a hypergraph of dimension at most $r$ and
$s$ of its minimal transversals, decide whether a further one exists.
Under ETH, no uniform algorithm solves this problem in time
$n^{f(r)}(s+1)^{o(r)}$. Consequently, no uniform algorithm enumerates
all minimal transversals in total time $n^{f(r)}(Q+1)^{o(r)}$, where
$Q$ is the total number of minimal transversals.
While this forbids a uniform output-polynomial bound whose exponent
in the number of solutions is independent of $r$, it does
not exclude even polynomial-delay enumeration for each fixed dimension $k$, and for instance one might still be able to find a polynomial-delay algorithm for hypergraphs of dimension at most say $1000$ without contradicting ETH.

\item Second, using the same reduction to fixed-width $d$-SAT, 
we derive a lower bound  based on the Strong Exponential Time Hypothesis (SETH) for the same problem. We prove that under SETH, for every fixed $k \geq 2$,
the problem of deciding whether a hypergraph has a minimal transversal of size at least $2k+1$ cannot be solved in time $n^{f(k)}m^{c}$ for any computable function $f$ and any constant $c<k$.
In consequence, for every fixed $k \geq 2$, deciding whether a supplied
family of $s$ minimal transversals is incomplete cannot be done in
time $\poly(|\cG|)(s+1)^c$ for any $c<k$, even when
$\hdim(\cG) \leq 2k$. As a corollary, no enumeration algorithm on this
class has total running time $\poly(|\cG|)(Q+1)^c$ for $c<k$.
This applies to any output-polynomial algorithm, without an incremental
time guarantee. In particular, for every fixed $k \geq  4$ no polynomial delay algorithm exists to enumerate minimal transversals of hypergraphs of dimension at most $k$.

\end{enumerate}

Capelli and Strozecki proved under ETH that polynomial-delay enumeration is
strictly less powerful than incremental-polynomial enumeration, using a
hierarchy of padded satisfiability problems
\cite{CapelliStrozecki2019}. However, a natural enumeration problem witnessing
this separation was still missing. To the best of our knowledge, our results
provide the first such example: under SETH, for every fixed dimension bound
$r \geq 4$, minimal transversals of hypergraphs of dimension at most $r$ can be
enumerated in incremental-polynomial time, but not with polynomial delay.

Another consequence is that already known incremental-polynomial time enumeration algorithms for bounded dimension and bounded VC-dimension are essentially optimal, in the sense that even for output-polynomial enumeration the exponent of the total number of solutions must be linear in the dimension and VC-dimension.

\section{Preliminaries and complexity conventions}

A hypergraph $\cH=(V,E)$ has $n=|V|$ vertices and $m=|E|$
(hyperedges). A transversal (or hitting set) $T\subseteq V$ is a subset of vertices that intersects every hyperedge,
it is \emph{minimal} if no proper subset is a transversal. We write
$\Tr(\cH)$ for the hypergraph whose hyperedges are the minimal
transversals of $\cH$. This hypergraph is called the \emph{transversal hypergraph} of $\cH$ or the \emph{dual} of $\cH$. The \emph{transversal rank} of a hypergraph $\cH$, noted $\trank(\cH)$, is defined as the dimension of its dual hypergraph, i.e. the maximum size of a minimal
transversal of $\cH$:
\[
  \trank(\cH)
    =\max\{|T|:T\in\Tr(\cH)\}.
\]
The \MMHS (MMHS) problem asks, given $(\cH,k)$, whether
$\trank(\cH) \geq  k$. In the literature this is also called
the \emph{upper transversal} problem or the \emph{transversal rank}
problem \cite{AraujoBougeretCamposSau2023,Schirneck2026}.

The \emph{dimension} $\hdim(\cH)$, also called the \emph{rank}
in the literature, is the largest cardinality of an input hyperedge..
 The VC-dimension $\VC(\cH)$ is the maximum size of a shattered vertex subset. In
particular,
\begin{equation}\label{eq:VCdimension}
 \VC(\cH) \leq \hdim(\cH).
\end{equation}
A hypergrpaph is \emph{Sperner} if no hyperedge contains another. 
Throughout this paper, except when explicitly stated otherwise, we assume that hypergraphs are Sperner.
We use $\cH_{\min}$ for the hypergraph of inclusion-minimal hyperedges
of $\cH$. It satisfies $\Tr(\cH_{\min})=\Tr(\cH)$,
A famous duality property of of the transversal hypergraph \cite{Berge1989} states that
\begin{equation}\label{eq:duality}
 \Tr(\Tr(\cH))=\cH_{\min}.
\end{equation}

\subsection{Enumeration Probelms}

An enumeration problem associates with every input $x$ a finite set
$\mathrm{Sol}(x)$ and asks to output every element of $\mathrm{Sol}(x)$ exactly
once. Since the number of solutions may be exponential in
the input size, the standard complexity measures are output-sensitive. An
algorithm runs in \emph{output-polynomial time} if its total running time is
polynomial in $|x|$ and in the total size of $\mathrm{Sol}(x)$. It runs in
\emph{incremental-polynomial time} if, for every $i$, it outputs $i$ distinct
solutions, or outputs all solutions and certifies and certify that they are fewer than $i$, in time polynomial in $|x|$ and $i$. Finally, it has
\emph{polynomial delay} if the time before the first output, between any two
consecutive outputs, and between the last output and termination is bounded by
a polynomial in $|x|$, independently of the number of solutions already
generated. These notions define the classes
$\mathrm{DelayP}\subseteq\mathrm{IncP}\subseteq\mathrm{OutputP}$
\cite{JohnsonPapadimitriouYannakakis1988,CapelliStrozecki2019}.

For our quantitative lower bounds, the \emph{new-transversal decision
problem} takes as input a hypergraph $\cG$ and an explicitly listed
family $\mathcal S\subseteq\Tr(\cG)$, and asks whether
$\Tr(\cG)\setminus\mathcal S\ne\varnothing$. Thus the supplied
sets are promised to be distinct minimal transversals. We write
$s=|\mathcal S|$ and measure decision time as
$p(|\cG|)(s+1)^c$, where $|\cG|$ denotes the explicit encoding size
of the hypergraph alone and $p$ is a polynomial.

For enumeration, we write $Q=|\Tr(\cG)|$. A \emph{total-time} bound
$p(|\cG|)(Q+1)^c$ includes all outputs and termination, but imposes
no bound on when individual solutions appear. Since every output has
at most $n$ vertices, counting solutions rather than their total
encoding size changes only the polynomial input-size factor.
These are quantitative restrictions on output-polynomial algorithms,
not claims that output-polynomial enumeration is impossible for a
fixed dimension. We also retain the delay convention: a bound
$p(|\cG|)(i+1)^\delta$ on the delay after $i$ solutions have been generated,
including the initial and final delays, implies, for $\delta \geq 0$, a total-time
bound $O(p(|\cG|)(Q+1)^{\delta+1})$. Polynomial delay is the special case
$\delta=0$ \cite{CapelliStrozecki2019}.

\subsection{ETH and SETH}

We use the standard Exponential Time Hypothesis (ETH), excluding a
$2^{o(N)}$-time algorithm for $3$-SAT on $N$ variables, and the Strong
Exponential Time Hypothesis (SETH). Let $s_d$ be the infimum over
$\alpha$ for which $d$-SAT can be decided in time
$2^{\alpha N}\poly(N)$. SETH states that $\sup_d s_d=1$
\cite{ImpagliazzoPaturi2001,ImpagliazzoPaturiZane2001}.

\section{Adaptation of the reduction of Kobayashi, Kurita and Wasa}

The following black-box lemma extracts from the construction and proofs
in \cite[Section~2, in particular Lemmas~3-7]{KKW2026} exactly the
properties needed below. Its statement for arbitrary fixed clause width
$d$ is a direct extension of the stated $3$-CNF construction: the proof of
the transversal characterization uses only the block selectors, the
original clause constraints, and the polynomial-evaluation labels.
Clause width is relevant only in bounding the number of original
constraints and works for any fixed size $d \geq 3$-clauses.

\begin{lemma}\label{lem:padding}
Fix integers $d \geq 3$ and $t \geq 1$. For all sufficiently large $N$,
there is a deterministic construction that maps any $d$-CNF formula
$\varphi$ on $N$ variables to a hypergraph
$\cH=\cH(\varphi,d,t)$ such that
\begin{equation}\label{eq:paddingcorrect}
 \varphi\in\SAT
    \quad\Longleftrightarrow\quad
       \trank(\cH) \geq 2t+1.
\end{equation}
Writing $n_{\cH}=|V(\cH)|$ and
$m_{\cH}=|E(\cH)|$, for fixed $d,t$,
\begin{align}
 n_{\cH}&=N^{O(dt)},\label{eq:paddingN}\\
 m_{\cH}& \leq  2^{\frac{N}{t}+o(N)}.\label{eq:paddingM}
\end{align}
The hypergraph can be constructed explicitly in time $2^{N/t+o(N)}$.
The $o(N)$ term may depend on the fixed integers $d,t$.
Here $\cH$ is not required to be Sperner; the MMHS lower bounds
below apply to arbitrary explicitly represented hypergraphs.
\end{lemma}

\begin{proof}
We extract the correctness statements from \cite[Lemmas~5-7]{KKW2026}.
If the formula contains an empty clause, output the hypergraph with one
vertex and one singleton hyperedge; its transversal rank is $1$, as required
for this unsatisfiable case.
After removing duplicate and tautological clauses, the number of clauses
of a $d$-CNF is $O_d(N^d)$. Let $q$ be a prime satisfying
$N^{d+2} \leq  q<2N^{d+2}$, for fixed $d$ a prime in this interval can be found
deterministically in $N^{O(d)}$ time. Set
\begin{equation}\label{eq:params}
 b=\lfloor t\log_2q\rfloor,
 \qquad s=  \left\lceil\frac{N}{b}\right\rceil.
\end{equation}
For fixed $d,t$ and sufficiently large $N$, the preliminary
hypergraph in \cite[Section~2.1]{KKW2026} has $s$ blocks of partial
assignments, each with at most $2^b \leq  q^t$ vertices, and at most
$M \leq  s+O_d(N^d) \leq  q$ preliminary hyperedges. For sufficiently large
$N$, we also have $1 \leq  b \leq  N$ and $2t<s$.
These are precisely the three hypotheses on the field size, label
degree, and number of blocks required by the padding construction
\cite[Section\;2.2]{KKW2026}.

After applying the same padding construction, by \cite[Lemma~5]{KKW2026}, any transversal meeting each block at most
once is a full selector representing a satisfying assignment, and every
such selector is minimal. By \cite[Lemma~6]{KKW2026}, any transversal
containing two vertices from one block has a transversal subset of
size at most $2t$. Therefore no minimal transversal of size greater
than $2t$ contains two vertices in the same block. The larger minimal
transversals are precisely the valid full selectors. Since such a
selector contains $s>2t$ vertices, this proves
\eqref{eq:paddingcorrect}.

The number of vertices is at most
$s2^b \leq (N+1)q^t=N^{O(dt)}$.
The padded hypergraph has at most $Mq^s$ hyperedges. Since
$b=t\log_2q+O(1)$ and $\log_2q=0(\log N)$,
\begin{align*}
 \log_2(Mq^s)
 & \leq  O(\log N)+\Bigl(\frac{N}{b}+1\Bigr)\log_2q\\
 &=\frac{N}{t}+O \left(\frac{N}{\log N}\right).
\end{align*}
Each padded hyperedge is generated by choosing a preliminary hyperedge and a vector
in $\mathbb F_q^s$ and evaluating the vertex labels, as in
\cite[Section~2.2]{KKW2026}. The work per hyperedge is polynomial in $N$
for fixed $d,t$, so the construction takes $2^{N/t+o(N)}$ time.
Duplicate hyperedges can be removed within the same bound by storing their
incidence vectors in a binary trie. No computation of the inclusion-minimal
edges is needed.
\end{proof}

\begin{remark}[Provenance and scope]
The explicit construction of $\cH$, including its padded
hyperedges, is not repeated here. The equivalence in
\eqref{eq:paddingcorrect} is the content of \cite[Lemma~7]{KKW2026}
under the admissible parameter conditions. 
\end{remark}

\section{ETH and SETH  hardness of maximum minimal hitting set}
\label{sec:mmhs}

Our first result treats the MMHS problem independently of enumeration.
It strengthens the usual $(n+m)^{\Omega(k)}$ ETH lower bounds
\cite{AraujoBougeretCamposSau2023,Schirneck2026} in an orthogonal
direction: it rules out moving \emph{all} parameter-dependent polynomial
complexity to the number of vertices. We then establish a SETH lower
bound for each fixed threshold $2t+1$, with $t \geq 2$. Both results
concern arbitrary explicitly represented hypergraphs, including those
that are not Sperner.

\begin{theorem}[ETH lower bound for MMHS]\label{thm:eth-mmhs}
Assume ETH. Let $a,f,g:\mathbb N\to\mathbb R_{ \geq 0}$ be computable,
with $g(k)=o(k)$. No uniform algorithm decides, on every $n$-vertex,
$m$-edge hypergraph, whether $\trank(\cH) \geq  k$ in time
\begin{equation}\label{eq:forbidden-mmhs}
 a(k)\,(n+1)^{f(k)}(m+1)^{g(k)}.
\end{equation}
In particular, the maximum minimal hitting set problem admits no running time
$n^{f(k)}m^{O(1)}$ under ETH, even for arbitrary computable $f$.
\end{theorem}

\begin{proof}
Assume an algorithm with the bound \eqref{eq:forbidden-mmhs} exists.
Fix an arbitrary integer $t \geq  1$. Apply Lemma~\ref{lem:padding}
with $d=3$ and this value of $t$ to a $3$-CNF formula with $N$
variables, and invoke the MMHS algorithm with $k=2t+1$.
The decision is correct by \eqref{eq:paddingcorrect}.
For every \emph{fixed} $t$, the total time, including the reduction,
is at most
\begin{align}
 2^{N/t+o(N)}
 +a(2t+1)N^{O(tf(2t+1))}
 2^{g(2t+1)(N/t+o(N))}
 =2^{\frac{g(2t+1)}{t}N+o(N)}.
 \label{eq:satT}
\end{align}
(The parameter-dependent constants and polynomial factors in $N$ are
absorbed into $o(N)$.) Since $g(k)=o(k)$,
\[
 \lim_{t\to\infty}
   \frac{g(2t+1)}t=0.
\]
For every $\varepsilon>0$, there is a \emph{fixed} integer $t$ for
which the algorithm \eqref{eq:satT} solves all sufficiently large
$3$-SAT instances within $2^{\varepsilon N}$ time.
By Theorem 2.4 of \cite{CHEN20092641}, this contradicts ETH.

\end{proof}

\begin{remark}[Relation to prior fine-grained bounds]
Schirneck \cite{Schirneck2026} gives a recognition algorithm of
running time\[
O(\Delta^{k-2}mn^{k-1}),
\]
where $\Delta \leq  m$ is the maximum vertex degree. Theorem~\ref{thm:eth-mmhs}
shows that the worst-case exponent of $m$ cannot be reduced to $o(k)$,
regardless of how rapidly the exponent of $n$ is allowed to grow with
$k$. It does not establish a particular positive constant in the
$\Omega(k)$ exponent of $m$.
\end{remark}

\begin{theorem}[SETH lower bound for MMHS]\label{thm:seth-mmhs}
Assume SETH. For every fixed integer $t \geq 2$ and every real constant
$c<t$, there is no algorithm deciding whether an $n$-vertex, $m$-edge
hypergraph $\cH$ satisfies $\trank(\cH) \geq 2t+1$ in time
\begin{equation}\label{eq:forbidden-seth-mmhs}
 a\,(n+1)^f(m+1)^c,
\end{equation}
where $a>0$ and $f \geq 0$ are arbitrary constants, allowed to depend on
$t$ and $c$.
\end{theorem}

\begin{proof}
Suppose such an algorithm exists and assume that $0 \leq  c<t$.
Fix any clause width $d \geq 3$ and apply Lemma~\ref{lem:padding} to a
$d$-CNF formula on $N$ variables. For the finitely many sizes at which
the lemma's parameter conditions fail, use brute force. Otherwise,
\eqref{eq:paddingcorrect} allows us to decide satisfiability by testing
$\trank(\cH) \geq 2t+1$. Including the construction of
$\cH$, the running time is bounded by
\[
 2^{N/t+o(N)}
 +a\,(N^{O(dt)}+1)^f
       \bigl(2^{N/t+o(N)}+1\bigr)^c
  \leq  2^{\rho N+o(N)},
  \rho=\frac{c}{t}<1.
\]
Choose $\varepsilon>0$ with $\rho+\varepsilon<1$. For every fixed
$d$, this gives a $d$-SAT algorithm with running time
$O(2^{(\rho+\varepsilon)N}\poly(N))$. Since $\rho+\varepsilon$ does not depend on $d$, 
$s_d \leq \rho+\varepsilon$ for every $d \geq 3$, and hence
$\sup_d s_d<1$, contradicting SETH.
\end{proof}


\section{The new-transversal decision problem: the uniform ETH bound}

We now transfer the MMHS lower bounds of Section~\ref{sec:mmhs} to
the new-transversal decision problem. The construction uses the small
minimal transversals of the input hypergraph, following the approach
underlying \cite[Theorem~4 and Section~5]{Schirneck2026}. We give a
direct proof that preserves the promise on the supplied family and
keeps its size at most the number of input hyperedges. We then derive
an enumeration corollary using only a total-time bound.

\begin{lemma}[A duality certificate]\label{lem:certificate}
Let $\cH$ be any hypergraph and
$\cL\subseteq\Tr(\cH)$. Then
\begin{equation}\label{eq:certificate}
 \cL\ne\Tr(\cH)
 \quad\Longleftrightarrow\quad
 \Tr(\cL)\setminus E(\cH)\ne\varnothing.
\end{equation}
\end{lemma}

\begin{lemma}[From new-transversal decision to MMHS]
\label{lem:bounded-output}
Let $r \geq 1$. Suppose the new-transversal decision problem on
dimension-at-most-$r$ hypergraphs with $n$ vertices and supplied
families of size $s$ can be solved in time $D_r(n,s)$, where
$D_r(n,s)$ is nondecreasing in $s$.
Then, for any $n$-vertex, $m$-edge hypergraph $\cH$, deciding
whether $\trank(\cH) \geq  r+1$ takes time
\begin{equation}\label{eq:transform-time}
 O\bigl((n+1)^{r+2}(m+1)\bigr)+D_r(n,m).
\end{equation}

\end{lemma}

\begin{proof}
Assume $1 \leq  r<n$ and that $\cH$ has
at least one hyperedge.
Test every vertex subset of size at most $r$ for minimal transversality
and form the Sperner hypergraph
\begin{equation}\label{eq:Ldefinition}
 \cL=\{S\in\Tr(\cH):|S| \leq  r\}.
\end{equation}
This can be done in time $O((n+1)^{r+2}(m+1))$. 
The dimension of $\cL$ is at most $r$, and its explicit encoding size is
$(n+1)^{r+O(1)}$, independently of $m$.

By definition, $\trank(\cH) \geq  r+1$ if and only if
$\cL\ne\Tr(\cH)$. Lemma~\ref{lem:certificate} turns
this into a test for an element of $\Tr(\cL)$ outside
$E(\cH)$.  If any hyperedge of $\cH$ is not a minimal transversal of $\cL$, we can directly answer no, so assume that $E(\cH)\subseteq\Tr(\cL)$.

Consequently, the new-transversal instance $(\cL,E(\cH))$
has answer \textsc{Yes} exactly when
$\trank(\cH) \geq  r+1$. Applying the assumed decision
algorithm proves \eqref{eq:transform-time}. 
\end{proof}

Now from  Lemma~\ref{lem:bounded-output} and Theorem~ \ref{thm:eth-mmhs} we obtain the following corollary.

\begin{corollary}[Uniform new-transversal lower bound under ETH]
\label{thm:eth-decision}
Assume ETH. There are no computable nonnegative functions
$A,F,G$ with $G(r)=o(r)$ and a uniform algorithm solving the
new-transversal decision problem on every dimension-at-most-$r$
hypergraph $\cG$ with $n$ vertices and every supplied family
$\mathcal S\subseteq\Tr(\cG)$ of size $s$ within
\[
 A(r)(n+1)^{F(r)}(s+1)^{G(r)}
\]
steps.
\end{corollary}

In particular, assuming ETH, no uniform algorithm can enumerate all minimal transversals of hypergraphs with a complexity having sublinear exponent in term of the dimension in the number of solutions.

\begin{corollary}[Uniform output-polynomial lower bound under ETH]
\label{cor:eth-output}
Assume ETH. There are no computable nonnegative functions $A,F,G$
with $G(r)=o(r)$ and a uniform algorithm that enumerates all minimal
transversals of every dimension-at-most-$r$ hypergraph on $n$ vertices
and terminates in total time
\[
 A(r)(n+1)^{F(r)}(Q+1)^{G(r)},
 \qquad Q=|\Tr(\cG)|.
\]
\end{corollary}

\begin{remark}[Why ETH does not settle any particular fixed dimension]
Theorem~\ref{thm:eth-decision} and Corollary~\ref{cor:eth-output}
concern a \emph{single uniform algorithm}
with an effective bound valid as $r$ varies. The transfer above uses
that algorithm at dimension bound $r=k-1$ for every MMHS parameter
$k$. An enumeration algorithm for just one fixed dimension bound
would instead yield an MMHS algorithm for just one fixed threshold,
which is not excluded by Theorem~\ref{thm:eth-mmhs}. In particular,
Corollary~\ref{cor:eth-output} does \emph{not} rule out a separate
polynomial-delay algorithm for a particular fixed dimension, including
dimension three or dimension four. This distinction motivates the use of SETH
in the next section.
\end{remark}

\section{SETH: lower bounds for every fixed even dimension}

Applying the same transfer to the fixed-threshold SETH lower bound
in Theorem~\ref{thm:seth-mmhs} gives a lower bound for every fixed
even dimension, first for the new-transversal decision problem and
then for total enumeration time. No further reduction from SAT is needed.

\begin{theorem}[SETH lower bound for new-transversal decision]
\label{thm:seth-fixed}
Assume SETH. Let $t \geq 2$ be any fixed integer. For every real
constant $c<t$, there is no algorithm that, given a hypergraph $\cG$ of dimension at most $2t$ and
a family $\mathcal S\subseteq\Tr(\cG)$ of size $s$, decides
whether $\Tr(\cG)\setminus\mathcal S\ne\varnothing$ in time
\begin{equation}\label{eq:forbidden-decision}
 \poly(|\cG|)\,(s+1)^c.
\end{equation}

\end{theorem}

\begin{proof}
Suppose such an algorithm exists for fixed $t \geq 2$ and $c<t$.

Given an arbitrary MMHS input $\cH$ with $n$ vertices and $m$
edges, construct
$\cL=\{S\in\Tr(\cH):|S| \leq 2t\}$ as in
Lemma~\ref{lem:bounded-output}. This is a dimension-at-most-$2t$
hypergraph with encoding size $(n+1)^{2t+O(1)}$. Since $t$ is fixed,
the input-size factor $\poly(|\cL|)$ in the assumed
decision bound is at most $C(n+1)^h$ for constants $C,h$ depending
only on the algorithm and $t$. The supplied family constructed in
Lemma~\ref{lem:bounded-output} is
$\mathcal S=E(\cH)\cap\Tr(\cL)$ and has size at most $m$.

Lemma~\ref{lem:bounded-output}, with $r=2t$, therefore gives an
algorithm deciding whether $\trank(\cH) \geq 2t+1$ in time
\[
 O\bigl((n+1)^{2t+2}(m+1)\bigr)
 +C(n+1)^h(m+1)^c
  \leq  C'(n+1)^{h'}(m+1)^{c},
\]
for suitable constants $C',h'$. Because $t \geq 2$ and $c<t$, this contradicts
Theorem~\ref{thm:seth-mmhs}.
\end{proof}

\begin{corollary}[SETH lower bound for output-polynomial enumeration]
\label{cor:seth-output}
Assume SETH. For every fixed integer $t \geq 2$ and every real
constant $c<t$, there is no algorithm that enumerates all minimal
transversals of every dimension-at-most-$2t$ hypergraph
$\cG$ and terminates in total time
\begin{equation}\label{eq:forbidden-enum}
 \poly(|\cG|)(Q+1)^c,
 \qquad Q=|\Tr(\cG)|.
\end{equation}

\end{corollary}

\begin{corollary}[Incremental delay and polynomial delay]
\label{cor:seth-delay}
Assume SETH. For fixed $t \geq 2$, no algorithm enumerates minimal
transversals of dimension-at-most-$2t$ hypergraphs with a bound
\[
 \poly(|\cG|)\,(i+1)^\delta
\]
on the $i$-th delay (including initial and final delays) whenever
$\delta<t-1$. In particular, \emph{polynomial-delay enumeration
is impossible for dimension at most four}, and hence for every fixed
dimension bound $r \geq 4$.
\end{corollary}


\begin{remark}[The dimension-three boundary]
Theorem~\ref{thm:seth-mmhs} starts at threshold $2t+1=5$ and
therefore yields an enumeration obstruction starting at dimension
$2t=4$. The transfer at dimension three would require an appropriate
MMHS lower bound at threshold four, which is not established here.
For $t=1$, even the construction time in Lemma~\ref{lem:padding}
has exponent $1$ in $N$, so it gives no SETH contradiction.
The method therefore does not exclude polynomial delay for
dimension-at-most-three hypergraphs. For dimension two, polynomial delay
is classical via minimal vertex covers and maximal independent sets
\cite{TsukiyamaEtAl1977,JohnsonPapadimitriouYannakakis1988}.
\end{remark}

\section{Discussions}
For dimension bounded by $r$, Eiter and Gottlob
\cite{EiterGottlob1995}, with subsequent improvements and analyses
\cite{EiterGottlobMakino2003,EiterMakinoGottlob2008}, established
incremental-polynomial enumeration of  $i^{r+2}n^{O(1)}$
delay that would imply a much faster MMHS algorithm.

For hypergraphs with VC-dimension $<h$, the procedure in
\cite[Theorem~1]{Mary2026}, given $i$
previously generated minimal transversals, finds a new one or
certifies completion in time
\begin{equation}\label{eq:vc-upper}
 O \left(2^h(ni)^{h+1}+n^{2h+1}i\right).
\end{equation}
For VC-dimension at most $d_0$, put $h=d_0+1$:
this is a next-solution bound whose exponent in the supplied-family
size $s=i$ is at most $d_0+2$. It also solves the new-transversal
decision problem within that bound. Summing over all $Q$ solutions
and including termination gives a total-time bound with exponent
at most $d_0+3$ in $Q+1$.
Our SETH lower bound excludes decision exponents in $s+1$ and
total enumeration exponents in $Q+1$ below
$\lfloor d_0/2\rfloor$ (and hence delay exponents below
$\lfloor d_0/2\rfloor-1$).
Thus the known dimension- and VC-dimension-dependent polynomial
exponents are optimal \emph{up to a constant factor asymptotically}.

\paragraph{Connection to Schirneck's equivalences.}
Theorem~4 of \cite{Schirneck2026} establishes quantitative
connections between deciding whether a large minimal hitting set exists,
recognizing $k$-conformality, and incremental enumeration of
maximal hypercliques and minimal hitting sets in bounded-dimension
hypergraphs. Our Theorem~\ref{thm:eth-mmhs} rules out the
MMHS running-time target $n^{f(k)}m^{o(k)}$ under ETH, and
Theorem~\ref{thm:eth-decision} rules out a uniform dimension-dependent
new-transversal decision exponent $o(r)$. By
Corollary~\ref{cor:eth-output}, the same obstruction holds for total
enumeration time, without an incremental-time assumption.
Theorem~\ref{thm:seth-fixed} and Corollary~\ref{cor:seth-output}
further give fixed-dimension SETH obstructions, which are not
consequences of ETH alone. These results are particularly useful
when comparing a single fixed-dimension enumeration algorithm to the
bounds in \cite{EiterGottlob1995,Mary2026}.

\bibliographystyle{alpha}
\bibliography{transversal_enumeration_lower_bounds}
\end{document}